\documentclass[letterpaper, 10 pt, conference]{ieeeconf}  

\IEEEoverridecommandlockouts                              
\usepackage{amsmath} 
\usepackage{amssymb}  

\usepackage{moreverb,url}
\usepackage[colorlinks,bookmarksopen,bookmarksnumbered,citecolor=red,urlcolor=red]{hyperref}
\usepackage{caption}
\usepackage{subcaption}
\usepackage{url}
\usepackage{amssymb}
\usepackage{latexsym}
\usepackage{psfrag}
\usepackage{epsfig}
\usepackage{graphicx}
\usepackage{epstopdf}
\usepackage{color}
\usepackage{tikz}
\usetikzlibrary{calc}

\newcommand\BibTeX{{\rmfamily B\kern-.05em \textsc{i\kern-.025em b}\kern-.08em
T\kern-.1667em\lower.7ex\hbox{E}\kern-.125emX}}

\newcommand{\R}{{\rm I\!R}}
\newcommand{\dfb}{\stackrel{\Delta}{=}}

\newtheorem{theorem}{\textbf{Theorem}}[section]

\newtheorem{definition}[theorem]{\textbf{Definition}}
\newtheorem{remark}[theorem]{Remark}

\title{\LARGE \bf
Distributed scaling control of rigid formations
}

\author{Hector Garcia de Marina, Bayu Jayawardhana and Ming Cao
\thanks{Hector Garcia de Marina is with the Ecole de la Aviation Civile, Toulouse, France. Bayu Jayawardhana and Ming Cao are with the Engineering and Technology Institute of Groningen, University of Groningen, 9747 AG Groningen, The Netherlands. (e-mail: \{h.j.de.marina, b.jayawardhana, m.cao\}@rug.nl). This work was supported by the the EU INTERREG program under the auspices of the SMARTBOT project and the work of Cao was also supported by the European Research Council (ERC-StG-307207) and the Dutch Technology Foundation (STW-vidi-14134).}}%

\begin{document}

\maketitle
\thispagestyle{empty}
\pagestyle{empty}

\begin{abstract}
Recently it has been reported that biased range-measurements among neighboring agents in the gradient distance-based formation control can lead to predictable collective motion. In this paper we take advantage of this effect and by introducing distributed parameters to the prescribed inter-distances we are able to manipulate the steady-state motion of the formation. This manipulation is in the form of inducing simultaneously the combination of constant translational and angular velocities and a controlled scaling of the rigid formation. While the computation of the distributed parameters for the translational and angular velocities is based on the well-known graph rigidity theory, the parameters responsible for the scaling are based on some recent findings in bearing rigidity theory. We carry out the stability analysis of the modified gradient system and simulations in order to validate the main result.
\end{abstract}

\section{INTRODUCTION}
The use of teams of autonomous agents has attracted a lot of interest in recent years. This is due to the fact that in many tasks, such as the transportation of objects or area exploration \& surveillance, robotic teams can effectively accomplish tasks with robustness against uncertain environment and offer new functionalities, e.g. enhanced sensing instrumentation \cite{sheng2006distributed}. One of the key tasks in coordinating a team of agents is the formation and motion control, where the former refers to keeping a prescribed shape while the later refers to the steering of it. In particular, a very active topic regarding formation control is the {\it{distance-based}} control for rigid shapes, where the combination of potential-gradient control and rigidity graph theory allows us to achieve (locally) a prescribed shape by only controlling the inter-distances between neighboring agents \cite{KrBrFr08, AnYuFiHe08}. It is a very appealing approach since the agents can work with only local information, such as their own frame of coordinates and the relative positions of their neighbors. In addition, the equilibrium at the prescribed shape is exponentially stable and it can be made robust against sensor's biases \cite{cao2007controlling,SuLiAn15,MarCaoJa15}. 

In this paper we propose a novel distributed control algorithm for achieving the following three tasks simultaneously:
\begin{enumerate}
	\item[i)] Formation \emph{scale-free} shape keeping.
	\item[ii)] Steering the scale-free formation as a whole with the combination of a constant translation velocity and a constant angular velocity applied at its centroid.
	\item[iii)] Precise scaling of the formation, i.e. controlling precisely the rate of growing or shrinking between two desired scaled shapes. The proposed control law even allows the changing between two different shapes.
\end{enumerate}

The findings of our work employ the recent results in \cite{zhao2014bearing} on {\it{bearing rigidity theory}}. Roughly speaking, bearing rigidity theory is employed for controlling a shape instead of focusing on maintaining constant distances or positions between neighbors, so one is interested in maintaining constant inner angles of the shape which can be obtained from the unit vectors between neighbors of a {\it{scale-free rigid shape}}. In fact these findings have been recently employed to control the translational motion of a rigid formation with a precise scaling rate in \cite{zhao2015translational}. 
The approach presented in this paper has several advantages over \cite{zhao2015translational}. Firstly, it does not require a common frame of coordinates for the agents. Secondly it is estimator free and it does not require global information such as the position of the centroid and its desired velocity. Lastly, the distance-based approach also allows rotational motion, a feature lost in the position-based control since the steady-state orientation is globally fixed by design.

The strategy employed in this paper is based on assigning motion parameters to the prescribed distances of a desired rigid formation. It has been reported in \cite{MouMorseBelSunAnd15} that when two neighboring agents differ in the prescribed distance to maintain, collective motion of the formation occurs. More precisely, the formation converges to a distorted version of the desired rigid shape and at the same time it undergoes a constant translation together with a rotation about its centroid as has been described in detail in \cite{sun2014formation}. It has been shown in \cite{MaJaCa15} that if these mismatches in the prescribed distance are taken as distributed motion parameters, one can maintain a desired rigid shape and at the same time control precisely a combination of a constant translation of the formation with a constant rotation about its centroid. By unifying the aforementioned results on motion control employing distributed motion parameters and bearing rigidity theory, one can control simultaneously the motion of the prescribed rigid shape and its scale in a precise way, i.e. not distorting a scale-free shape for a desired rate of growing/shrinking. 

The rest of the paper is organized as follows. In Section \ref{sec: pre} we introduce the notation and background for bearing rigid formations. Section \ref{sec: mis} explains the design of the motion controller with precise scaling/morphing of the formation by introducing changing-motion parameters in a distance-based controller. In Section \ref{sec: sta} we demonstrate the exponential convergence of our proposed algorithm. Numerical simulations validate the main results in this paper in Section \ref{sec: exp}.

\section{PRELIMINARIES}
\label{sec: pre}
We start by introducing some notation employed throughout the paper. For a given matrix $A\in\R^{n\times p}$, define $\overline A \dfb A \otimes I_m \in\R^{nm\times pm}$, where the symbol $\otimes$ denotes the Kronecker product, with $m = 2$ for the 2D formation case or $m=3$ for the 3D case, and $I_m$ is the $m$-dimensional identity matrix. For a stacked vector/matrix $x\dfb \begin{bmatrix}x_1^T & x_2^T & \dots & x_k^T\end{bmatrix}^T$ with $x_i\in\R^{n\times l}, i\in\{1,\dots,k\}$, we define the block diagonal matrix $D_x \dfb \operatorname{diag}\{x_i\}_{i\in\{1,\dots,k\}} \in\R^{kn\times kl}$. We denote by $|\mathcal{X}|$ the cardinality of the set $\mathcal{X}$, by $||x||$ the Euclidean norm of a vector $x$ and by $\hat x=\frac{x}{||x||}$ the unit vector of a non-zero $x$. We define the orthogonal projector operator as $P^{\perp}_x \dfb \left(I_m - \hat x_k\hat x_k^T\right)$ or more generally $P^{\perp}_\mathcal{X}$ over an orthogonal subspace of $\mathcal{X}$. Finally we use $\mathbf{1}_{n\times m}$ and $\mathbf{0}_{n\times m}$ to denote the all-one and all-zero matrix in $\R^{n\times m}$ respectively and will drop the subscript if the dimensions are clear from the context.
\subsection{Graphs and rigidity theory}
We consider a formation of $n\geq 2$ agents whose positions are denoted by $p_i\in\R^m$ for $i\in\{1,\dots,n\}$. The agents can measure their range and directions with respect to their neighbors. The representation of this sensing topology is given by an undirected graph $\mathbb{G} = (\mathcal{V}, \mathcal{E})$ with the vertex set $\mathcal{V} = \{1, \dots, n\}$ and the ordered edge set $\mathcal{E}\subseteq\mathcal{V}\times\mathcal{V}$. The set $\mathcal{N}_i$ of the neighbors of agent $i$ is defined by $\mathcal{N}_i\dfb\{j\in\mathcal{V}:(i,j)\in\mathcal{E}\}$. We define the elements of the incidence matrix $B\in\R^{|\mathcal{V}|\times|\mathcal{E}|}$ for $\mathbb{G}$ by
\begin{equation*}
	b_{ik} \dfb \begin{cases}+1 \quad \text{if} \quad i = {\mathcal{E}_k^{\text{tail}}} \\
		-1 \quad \text{if} \quad i = {\mathcal{E}_k^{\text{head}}} \\
		0 \quad \text{otherwise,}
	\end{cases}
\end{equation*}
where $\mathcal{E}_k^{\text{tail}}$ and $\mathcal{E}_k^{\text{head}}$ denote the tail and head nodes, respectively, of the edge $\mathcal{E}_k$, i.e. $\mathcal{E}_k = (\mathcal{E}_k^{\text{tail}},\mathcal{E}_k^{\text{head}})$. Since the graph $\mathbb{G}$ is undirected, it is irrelevant how the directions of the edges are defined in $B$. 

A \emph{framework} is defined by the pair $(\mathbb{G}, p)$, where $p = \begin{bmatrix}p_1^T & \dots & p_n^T\end{bmatrix}^T$ is the stacked vector of the agents' positions. The available relative positions of the agents in the framework are given by the following stacked vector
\begin{equation*}
	z = \overline B^Tp,
\end{equation*}
where each vector $z_k = p_i - p_j$ in $z$ corresponds to the relative position associated with the edge $\mathcal{E}_k = (i, j)$.

Let us now briefly recall the notions of \emph{distance} infinitesimally rigid framework and minimally rigid framework from  \cite{AnYuFiHe08}. Define the edge function $f_\mathbb{G}$ by $f_{\mathbb{G}}(p) = \mathop{\text{col}}\limits_{k}\big(\|z_k\|^2\big)$ where the operator $\text{col}$ defines the stacked column vector and we denote its Jacobian, also known as the \emph{rigidity matrix}, by $R(z) = D_z^T\overline B^T$. A framework $(\mathbb{G}, p)$ is {\it infinitesimally rigid} if $\text{rank} R(z) = 2n-3$ when it is embedded in $\mathbb{R}^2$ or if $\text{rank} R(z) = 3n-6$ when it is embedded in $\mathbb{R}^3$. Additionally, if $|\mathcal E|=2n-3$ in the 2D case or $|\mathcal E|=3n-6$ in the 3D case then the framework is called {\it minimally rigid}. Roughly speaking, the only motions that one can perform over the agents in a minimally rigid framework, while they are already in the desired shape, are the ones defining translations and rotations of the whole shape. 

The stacked vector of relative positions $z^* = [\begin{smallmatrix}{z_1^*}^T & {z_2^*}^T & \dots & {z_{|\mathcal{E}|}^*}^T\end{smallmatrix}]^T$ defines a desired infinitesimally and minimally rigid shape with $||z_k^*|| = d_k$ for all $k\in\{1,\dots,|\mathcal{E}|\}$ where $d_k$ is the desired inter-distance. The resulting set $\mathcal{Z}$ of the possible formations with the same shape is defined by
\begin{equation}
	\mathcal{Z} \dfb \left \{\left(I_{|\mathcal{E}|} \otimes \mathcal{R}\right)z^* \right \}, \label{eq: Z}
\end{equation}
where $\mathcal{R}$ is the set of all rotational matrices in 2D or 3D. Roughly speaking, $\mathcal{Z}$ consists of all formation positions that are obtained by rotating $z^*$. 

Consider a scale-free shape \emph{based on} an infinitesimally and minimally rigid shape, for example the collection of all regular squares with an internal diagonal. It is obvious that this collection can be distinguished from other (infinitesimally and minimally rigid) scale-free shapes by looking at its inner angles or equivalently by looking at all the scalar products $\hat z_l^T\hat z_n$ where $l$ and $n$ are two edges sharing a node. This fact has been explained in more detail in \cite{6862558}. Bearing-based rigid frameworks are related to the distance-based ones where the bearing-based shape can be defined by the inner angles, instead of the distances. Let us review some basic concepts in bearing rigidity.
\begin{definition}\cite{zhao2014bearing}
	Frameworks $(\mathcal{G}, p)$ and $(\mathcal{G}, p')$ are \emph{bearing equivalent} if $P^{\perp}_{z_k} z_k' = 0$ for all $k\in\{1,\dots,|\mathcal{E}|\}$.
\end{definition}
\begin{definition}\cite{zhao2014bearing}
	Frameworks $(\mathcal{G}, p)$ and $(\mathcal{G}, p')$ are \emph{bearing congruent} if $P^{\perp}_{(p_i - p_j)} (p_i' - p_j') = 0$ for all $i, j\in \mathcal{V}$.
\end{definition}
\begin{definition}
	\label{def: bf}
	\cite{zhao2014bearing} The \emph{bearing function} is defined by $f_{B_\mathbb{G}}(p) \dfb \hat z \in\R^{m|\mathcal{E}|}$, where\footnote{In order not to overload the notation, here by $\hat z$ we mean exclusively the vector-element wise normalization of $z$.} $\hat z$ is the stacked vector of $\hat z_k$ for all $k\in\{1,\dots,|\mathcal{E}|\}$.
\end{definition}
Similar to the rigidity matrix one can define the \emph{bearing rigidity matrix} by computing the Jacobian matrix of the bearing function
\begin{equation}
	R_B(z) = \frac{\partial f_{B_\mathbb{G}}(p)}{\partial p} = \overline D^T_{\tilde z} D^T_{P^{\perp}_{\hat z}}\overline B^T, \nonumber
\end{equation}
where $P^{\perp}_{\hat z}\in\R^{m|\mathcal{E}|\times m}$ is the stacked matrix of operators $P^{\perp}_{\hat z_k}$ and $\tilde z\in\R^{|\mathcal{E}|}$ is the stacked vector of $\frac{1}{||z_k||}$ for all $k\in\{1,\dots,|\mathcal{E}|\}$. The non-trivial kernel of $R_B(z)$ includes the scalings and translations of the framework \cite{zhao2014bearing}, leading to the following definition
\begin{definition}
	\cite{zhao2014bearing} A framework is \emph{infinitesimally bearing rigid} if the kernel of its bearing rigidity matrix only includes scalings and translations.
\end{definition}
In order words, if a scale-free shape can be determined uniquely by its inner angles, then it belongs to the \emph{infinitesimally bearing rigid} framework. 

Consider a given shape defined by $\mathcal{Z}$, we define the scale-free $\mathcal{Z}_\mathcal{S}$ by taking $\mathcal{Z}$ rescaled by all the possible scale factors $s\in\R^+$ such that $||z_k|| = sd_k$ for all $k\in\{1,\dots,|\mathcal{E}|\}$. This leads to the following definition

\begin{definition}
	The shapes defined by $\mathcal{Z}$ within the set $\mathcal{Z}_\mathcal{S}$ are {\it{infinitesimally and minimally congruent}} rigid.
\end{definition}

The name comes from the fact that all the scales of an infinitesimally and minimally rigid shape are bearing congruent. 

\subsection{Frames of coordinates}
In order to describe and design motions for the desired scale-free formation defined by $\mathcal{Z}_\mathcal{S}$, it will be useful to attach a frame of coordinates to the centroid of the shape. We denote by $O_g$ the \emph{global frame} of coordinates fixed at the origin of $\R^m$ with some arbitrary fixed orientation. In a similar way, we denote by $O_b$ the \emph{body frame} fixed at the centroid $p_c$ of the desired scale-free rigid formation. Furthermore, if we rotate the scale-free rigid formation with respect to $O_g$, then $O_b$ is also rotated in the same manner. Note that $p_c$ is invariant with respect to $\mathcal{Z}_\mathcal{S}$. Let $^bp_i$ denote the position of agent $i$ with respect to $O_b$. In order to simplify notation, whenever we represent an agent's variable with respect to $O_g$, the superscript is omitted, e.g. $p_i \dfb {^gp_i}$.

\section{Motion and scaling of rigid formations}
\label{sec: mis}
We consider the $n$ agents in the framework $(\mathbb{G}, p)$ to be governed by single integrator dynamics
\begin{equation}
	\dot p_i = u_i, \label{eq: pdyn}
\end{equation}
where $u_i\in\R^m$ is the control action for all $i\in\{1, \dots, |\mathcal{V}|\}$. 
For each edge $\mathcal{E}_k$ in the framework one can associate a potential function $V_k(z_k)$ whose minimum corresponds to the desired configuration of the associated edge, for example, in order to (locally) stabilize $\mathcal{Z}$ we can employ the classical \emph{elastic potential function} from physics for controlling the length of the edges
\begin{equation}
	V(p) = \sum_{i=1}^{|\mathcal{E}|}V_k(z_k) = \frac{1}{2}\sum_{i=1}^{|\mathcal{E}|}(||z_k|| - d_k)^2.  \label{eq: vkd}
\end{equation}

It has been reported in \cite{MouMorseBelSunAnd15} that in undirected gradient-based controlled formations if at least two neighboring agents differ about the prescribed distance to maintain, i.e. they have a mismatch, then a steady-state collective motion with a distorted shape occurs. The collective motion, illustrated in Figure \ref{fig: 3dmotion}, is described by the combination of two constant velocity vectors: 
\begin{itemize}
	\item A linear velocity $^bv_{p_c}^*$ of the centroid with respect to the steady-state distorted shape.
	\item An angular velocity $^b\omega^*$ that rotates the steady-state distorted shape.
\end{itemize}
\begin{figure}
\centering
\begin{tikzpicture}[line join=round]
[\tikzset{>=latex}]\filldraw[fill=white](-2.353,1.014)--(-1.808,-.927)--(2.955,-.705)--(2.41,1.236)--cycle;
\filldraw[draw=black,fill=white,fill opacity=0.8](.55,.197)--(.958,.598)--(.958,.598)--(1.638,.005)--cycle;
\filldraw[draw=black,fill=white,fill opacity=0.8](.685,-.283)--(.958,.598)--(.958,.598)--(.55,.197)--cycle;
\draw[draw=red,arrows=->](2.381,.111)--(2.381,1.112);
\draw[draw=black,arrows=<->,thick](-.368,-.243)--(-.942,.059)--(-.569,.51);
\draw[arrows=-,thick](.992,.009)--(.958,.129)--(1.468,.153);
\draw[draw=black,arrows=-](.958,.129)--(.958,.598);
\draw[draw=black,arrows=-](.958,.129)--(1.148,.279);
\draw[arrows=-,thick](.958,.129)--(.958,.598);
\draw[draw=black,arrows=-](.958,.129)--(1.309,.01);
\filldraw[draw=black,fill=white,fill opacity=0.8](1.638,.005)--(.958,.598)--(.958,.598)--(.685,-.283)--cycle;
\draw[arrows=->,thick](1.468,.153)--(2.115,.183);
\draw[draw=black,arrows=->](1.309,.01)--(1.847,-.172);
\draw[draw=black,arrows=-](1.148,.279)--(1.378,.46);
\draw[thick](1.162,1.08)--(1.167,1.082)--(1.171,1.083)--(1.176,1.085)--(1.18,1.086)--(1.185,1.088)--(1.189,1.09)--(1.194,1.091)--(1.198,1.093)--(1.202,1.095)--(1.206,1.097)--(1.211,1.098)--(1.215,1.1)--(1.218,1.102)--(1.222,1.104)--(1.226,1.106)--(1.23,1.108)--(1.233,1.11)--(1.237,1.112)--(1.24,1.114)--(1.244,1.116)--(1.247,1.118)--(1.25,1.12)--(1.253,1.122)--(1.256,1.124)--(1.259,1.126)--(1.262,1.128)--(1.264,1.13)--(1.267,1.133)--(1.27,1.135)--(1.272,1.137)--(1.274,1.139)--(1.277,1.141)--(1.279,1.144)--(1.281,1.146)--(1.283,1.148)--(1.285,1.151)--(1.286,1.153)--(1.288,1.155)--(1.29,1.158)--(1.291,1.16)--(1.292,1.162)--(1.294,1.165)--(1.295,1.167)--(1.296,1.17)--(1.297,1.172)--(1.297,1.174)--(1.298,1.177)--(1.299,1.179)--(1.299,1.182)--(1.3,1.184)--(1.3,1.187)--(1.3,1.189)--(1.3,1.191)--(1.3,1.194)--(1.3,1.196)--(1.3,1.199)--(1.3,1.201)--(1.299,1.204)--(1.299,1.206)--(1.298,1.208)--(1.297,1.211)--(1.297,1.213)--(1.296,1.216)--(1.295,1.218)--(1.293,1.22)--(1.292,1.223)--(1.291,1.225)--(1.289,1.228)--(1.288,1.23)--(1.286,1.232)--(1.284,1.235)--(1.283,1.237)--(1.281,1.239)--(1.279,1.242)--(1.277,1.244)--(1.274,1.246)--(1.272,1.248)--(1.269,1.251)--(1.267,1.253)--(1.264,1.255)--(1.262,1.257)--(1.259,1.259)--(1.256,1.261)--(1.253,1.263)--(1.25,1.266)--(1.247,1.268)--(1.243,1.27)--(1.24,1.272)--(1.237,1.274)--(1.233,1.276)--(1.23,1.278)--(1.226,1.28)--(1.222,1.281)--(1.218,1.283)--(1.214,1.285)--(1.21,1.287)--(1.206,1.289)--(1.202,1.29)--(1.198,1.292)--(1.194,1.294)--(1.189,1.296)--(1.185,1.297)--(1.18,1.299)--(1.176,1.3)--(1.171,1.302)--(1.166,1.303)--(1.162,1.305)--(1.157,1.306)--(1.152,1.308)--(1.147,1.309)--(1.142,1.31)--(1.137,1.312)--(1.132,1.313)--(1.126,1.314)--(1.121,1.315)--(1.116,1.316)--(1.111,1.318)--(1.105,1.319)--(1.1,1.32)--(1.094,1.321)--(1.089,1.322)--(1.083,1.322)--(1.078,1.323)--(1.072,1.324)--(1.066,1.325)--(1.061,1.326)--(1.055,1.326)--(1.049,1.327)--(1.044,1.328)--(1.038,1.328)--(1.032,1.329)--(1.026,1.329)--(1.02,1.33)--(1.014,1.33)--(1.008,1.331)--(1.002,1.331)--(.997,1.331)--(.991,1.332)--(.985,1.332)--(.979,1.332)--(.973,1.332)--(.967,1.332)--(.961,1.332)--(.955,1.332)--(.949,1.332)--(.943,1.332)--(.937,1.332)--(.931,1.332)--(.925,1.332)--(.919,1.331)--(.913,1.331)--(.907,1.331)--(.901,1.33)--(.895,1.33)--(.89,1.329)--(.884,1.329)--(.878,1.328)--(.872,1.328)--(.866,1.327)--(.861,1.326)--(.855,1.326)--(.849,1.325)--(.843,1.324)--(.838,1.323)--(.832,1.322)--(.827,1.322)--(.821,1.321)--(.816,1.32)--(.81,1.319)--(.805,1.317)--(.8,1.316)--(.794,1.315)--(.789,1.314)--(.784,1.313)--(.779,1.312)--(.774,1.31)--(.769,1.309)--(.764,1.308)--(.759,1.306)--(.754,1.305)--(.749,1.303)--(.745,1.302)--(.74,1.3)--(.735,1.299)--(.731,1.297)--(.726,1.295)--(.722,1.294)--(.718,1.292)--(.714,1.29)--(.709,1.289)--(.705,1.287)--(.701,1.285)--(.697,1.283)--(.694,1.281)--(.69,1.279)--(.686,1.277)--(.683,1.276)--(.679,1.274)--(.676,1.272)--(.672,1.27)--(.669,1.268)--(.666,1.265)--(.663,1.263)--(.66,1.261)--(.657,1.259)--(.654,1.257)--(.651,1.255)--(.649,1.253)--(.646,1.25)--(.644,1.248)--(.642,1.246)--(.639,1.244)--(.637,1.241)--(.635,1.239)--(.633,1.237);
\draw[draw=black,dashed](1.847,.829)--(1.847,-.172);
\draw[arrows=->,thick](-.942,.059)--(-.932,.476);
\draw[draw=black,arrows=->](.958,.598)--(.958,.942);
\draw[arrows=->,thick](.958,.598)--(.958,1.193);
\draw[arrows=<-,thick](1.09,-.342)--(.992,.009);
\draw[draw=black,arrows=->](1.378,.46)--(1.847,.829);
\draw[draw=black,dashed](1.847,.829)--(.958,.942);
\draw[thick,->](1.206,1.097)--(1.055,1.059);
\node at (-1.233,.288) {$O_g$};\node at (.428,.484) {$O_b$};\node at (.958,1.568) {$^b\omega^*$};\node at (2.012,.878) {$^bv^*_{p_c}$};\draw[red,->] (.958,.129) to [bend left=45] (1.631,-.897);\end{tikzpicture}
	\caption{The resultant motion of the tetrahedron is the composition of the two constant velocities $^bv_{p_c}^*$ and $^b\omega^*$. The described trajectory has been split in two red curves.} 
\label{fig: 3dmotion}
\end{figure}
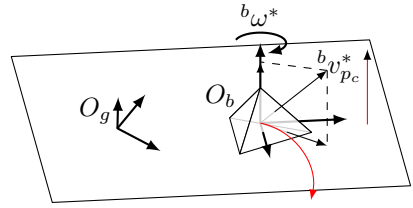

It has been shown in \cite{MaJaCa15} that if we replace mismatches by distributed motion parameters, then we can control both, a non-distorted desired shape and a desired motion of the formation with respect to $O_b$. Throughout this section we will show that such an approach can also be employed to scale the formation shape precisely over time while simultaneously to guide the formation to travel with a desired $^bv_{p_c}^*$ and $^b\omega^*$. This approach is easier and more effective in several aspects than the one presented in \cite{zhao2015translational}, since we do not need estimators for traveling at a constant speed and we can also rotate the desired formation with respect to $O_g$ by controlling only one agent. Moreover, using our proposed approach, the agents can use only local coordinates since we are using the distance-based control strategy.

We introduce the motion and changing parameters to the gradient-based control and show the steady-state motion, including the scaling of the shape $\mathcal{Z}$, is related to its unit vectors. The control inputs derived from the gradient of the distance-based potential (\ref{eq: vkd}) for the agents $i$ and $j$ on the edge $\mathcal{E}_k = (i, j)$ are as follows
\begin{equation}
\left.
\begin{array}{ll} u_i^k &= -\hat z_k\big(||z_k|| - d_k\big) \\
		u_j^k &= \hat z_k\big(||z_k|| - d_k\big),
\end{array}\right\}
\end{equation}
where the superscript $k$ denotes the contribution of the edge $k$ to the total control input $u_i$ and $u_j$. Introduce a pair of parameters $\mu_k$ and $\tilde\mu_k$ to the prescribed distance $d_k$ as follows
\begin{equation}
\left.
\begin{array}{ll}
	u_i^k = -\hat z_k\big(||z_k|| - d_k - \mu_k\big) &= -\hat z_k\big(||z_k|| - d_k) +  \hat z_k\mu_k\\
	u_j^k = \hat z_k\big(||z_k|| - d_k + \tilde\mu_k\big) &= \hat z_k\big(||z_k|| - d_k) + \hat z_k\tilde\mu_k.
\end{array}\right\} \label{eq: mis}
\end{equation}
The structure in (\ref{eq: mis}) allows us to write the complete control law $u$ in the following compact form
\begin{equation}
	u = -c \overline BD_{\hat z}e + \overline A(\mu, \tilde\mu)\hat z, \label{eq: umis}
\end{equation}
where $u\in\R^{m|\mathcal{V}|}$ is the stacked vector of control actions $u_i$, $c\in\R^+$ is a constant gain, $e\in\R^{|\mathcal{E}|}$ is the stacked vector of all the distance errors $e_k = ||z_k|| - sd_k$ where all the $sd_k$'s have been taken from $\mathcal{Z}_s$, the parameters $\mu\in\R^{|\mathcal{E}|}$ and $\tilde\mu\in\R^{|\mathcal{E}|}$ are the stacked vectors of $\mu_k$ and $\tilde\mu_k$ for all $k\in\{1,\dots,|\mathcal{E}|\}$ and the elements of $A$ are defined as follows
\begin{equation}
a_{ik} \dfb \begin{cases}\mu_k \quad \text{if} \quad i = {\mathcal{E}_k^{\text{tail}}} \\
		\tilde\mu_k \quad \text{if} \quad i = {\mathcal{E}_k^{\text{head}}} \\
		0 \quad \text{otherwise.}
	\end{cases} \label{eq: A}
\end{equation}
Note that the elements of $A$ are related to the incidence matrix $B$ because of (\ref{eq: mis}), and hence we still have a distributed control law. 

We can identify two terms at the right hand side of (\ref{eq: umis}). The first one is clearly related to the gradient distance-based controller and its purpose is to form and keep the prescribed shape given by $\mathcal{Z}_s$. The second term corresponds to the steady-state collective motion and changing induced by the parameters $\mu_k$ and $\tilde\mu_k$ and the actual shape of the formation given by the unit vectors in $\hat z$. We will see that in order to guarantee the stability of the system we will make use of the exponential convergence of the self-contained error system in the original gradient-based controller. By choosing $c$ in (\ref{eq: umis}) sufficiently large, we can make the gradient-based term dominate the second term. Therefore the team of agents will converge to the desired shape $\mathcal{Z}_s$, where $s$ can be time-varying i.e. we will scale the shape within $\mathcal{Z}_\mathcal{S}$, and the steady-state motion will be given by the parameters and the unit vectors in $z\in\mathcal{Z}_\mathcal{S}$.

\subsection{Design of the distributed motion and changing parameters}
Suppose that the formation is at the prescribed shape, i.e. $e=\mathbf{0}$. In this case if $\overline A(\mu, \tilde\mu)\hat z$ defines translations and rotations of the infinitesimally and minimally congruent rigid family $\mathcal{Z}_\mathcal{S}$, then the desired scaled shape $\mathcal{Z}_s$ will be invariant under such an additional control term. Note that from (\ref{eq: umis}) when $e=\mathbf{0}$ the control law for the agent $i$ becomes
\begin{equation}
	^bu_i = \sum_{k=1}^{|\mathcal{E}|}a_{ik}{^b\hat z}_k^*, \label{eq: uim}
\end{equation}
where $^b\hat z^*\in\mathcal{Z}_{\mathcal{S}}$. We recall that the elements $a_{ik}$ of $A$ are related to $\mu$ and $\tilde\mu$ as in (\ref{eq: A}). In an infinitesimally and minimally congruent rigid formation, the minimum number of neighbors for agent $i$ is two (resp. three) in 2D (resp. 3D) shapes with its corresponding $z_k^*$'s not being in non-generic degenerated configurations, e.g. all of them collinear (resp. coplanar), then $^bu_i$ can span the whole $\R^2$ (resp . $\R^3$). In other words, we can design a pair of arbitrary constant velocities $^bv_{p_c}^*$ and $^b\omega^*$ for the desired scale-free formation $\mathcal{Z}_{\mathcal{S}}$ by choosing appropriately $\mu$ and $\tilde\mu$. For choosing such $\mu$ and $\tilde\mu$, let us decompose them into $\mu = \mu_v + \mu_\omega + \mu_s$ and $\tilde\mu = \tilde\mu_v + \tilde\mu_\omega + \tilde\mu_s$, where each term in this decomposition can be used to define the translation, rotation and scaling of the group motion. Here, the subscript `$v$' refers to the motion parameters responsible for $^bv_{p_c}^*$, `$\omega$' refers to $^b\omega^*$ and finally `$s$' refers to the changing parameters which are responsible for dilating/contracting the shape within $\mathcal{Z}_{\mathcal{S}}$. As shown in \cite{MaJaCa15} the motion parameters of $\mu_v, \tilde\mu_v, \mu_\omega$ and $\tilde\mu_\omega$ can be determined by imposing restrictions on the dynamics of $z$ and $e$ in order to keep invariant $e$ at $e = \mathbf{0}$, i.e.
\begin{align}
	\overline B^T\overline A(\mu, \tilde\mu)\hat z = 0 \label{eq: Z} \\
	D_{\hat z}\overline B^T\overline A(\mu, \tilde\mu)\hat z = 0, \label{eq: E}
\end{align}
where (\ref{eq: Z}) stands for translations and (\ref{eq: E}) stands for rotations and translations. Let us write the following identity
\begin{equation}
	A(\mu, \tilde\mu)\hat z = \begin{bmatrix}\bar S_1 D_{\hat z} & \bar S_2D_{\hat z}\end{bmatrix}\begin{bmatrix}\mu \\ \tilde\mu\end{bmatrix} = T(\hat z)\begin{bmatrix}\mu \\ \tilde\mu\end{bmatrix}, \label{eq: T}
\end{equation}
where $S_1$ is constructed by setting all the $1$ elements in the incidence matrix $B$ to zero and $S_2 \dfb S_1 - B$. In order to compute the distributed motion parameters $\mu_v, \tilde\mu_v$ for the translational velocity $^bv_{p_c}^*$ we eliminate the components of $\mu$ and $\tilde\mu$ responsible of non-motion, i.e. $A(\mu, \tilde\mu)\hat z = 0$ in (\ref{eq: umis}), by projecting the kernel of $\overline B^TT(^b\hat z^*)$ (derived from (\ref{eq: Z})) over the orthogonal space of the kernel of (\ref{eq: T})
\begin{equation}
	\begin{bmatrix}\mu_v \\ \tilde\mu_v \end{bmatrix} \in \mathcal{\hat U}\dfb P_{\operatorname{Ker}\{T(^b\hat z^*)\}}^\bot\left\{\operatorname{Ker}\{\overline B^TT(^b\hat z^*)\} \right\}	\label{eq: U}
\end{equation}
where we have employed $^b\hat z^*$ in order to define the translational velocity of the desired formation with respect to $O_b$ as in Figure \ref{fig: 3dmotion}. In a similar way, by removing the components responsible of non-motion and translational velocity, the computation of the distributed motion parameters $\mu_\omega, \tilde\mu_\omega$ for the rotational motion of the desired shape is obtained from (\ref{eq: E}) and (\ref{eq: U}) as
\begin{equation}
	\begin{bmatrix}\mu_\omega \\ \tilde\mu_\omega \end{bmatrix} \in  \mathcal{\hat W}\dfb P_{\mathcal{\hat U}}^\bot\left\{\operatorname{Ker}\{D_{^b\hat z^*}^T\overline B^TT(^b\hat z^*)\}\right\}.
\label{eq: W}
\end{equation}
In order to compute the distributed changing parameters $\mu_s, \tilde\mu_s$ we need to look at the bearing rigidity matrix $R_B(z)$. It has been shown in \cite{zhao2014bearing} that the meaning of the kernel of the bearing rigidity matrix stands for translations and scalings of the desired shape. Therefore in a similar way as before the following condition
\begin{equation}
	\begin{bmatrix}\mu_s \\ \tilde\mu_s\end{bmatrix} \in \mathcal{S}\dfb P_{\mathcal{\hat U}}^\bot\left\{\operatorname{Ker}\{D^T_{P^{\perp}_{^b\hat z^*}}\overline B^TT(^b\hat z^*)\}\right\},
	\label{eq: S}
\end{equation}
will give us the space of changing parameters responsible for the scaling of the formation. We would like to remark that the presented motion and changing parameters have been designed for a family of infinitesimally and minimally congruent shapes $\mathcal{Z}_{\mathcal{S}}$, i.e. for a scale-free version of a desired infinitesimally and minimally rigid shape. Note that the three spaces $\mathcal{U}, \mathcal{W}$ and $\mathcal{S}$ have been computed in a centralized way while the parameters $\mu$ and $\tilde\mu$ are applied in a distributed fashion. This computation can be done \emph{off-line} during the design stage, e.g. at the same time that one designs the corresponding $d_k$ for the desired shape $\mathcal{Z}$. An example of the three spaces $\mathcal{U}, \mathcal{W}$ and $\mathcal{S}$ for an infinitesimally and minimally congruent regular squares is given in Figure \ref{fig: mis}.
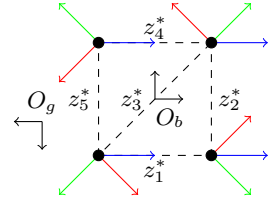
\begin{figure}
\centering
\begin{tikzpicture}[line join=round]
\draw[dashed](0,0)--(1.5,0)--(1.5,1.5)--(0,1.5)--(0,0)--(1.5,1.5);
\draw[draw=red,arrows=->](0,0)--(.525,-.525);
\draw[draw=red,arrows=->](1.5,0)--(2.025,.525);
\draw[draw=red,arrows=->](1.5,1.5)--(.975,2.025);
\draw[draw=red,arrows=->](0,1.5)--(-.525,.975);
\draw[draw=blue,arrows=->](0,0)--(.75,0);
\draw[draw=blue,arrows=->](1.5,0)--(2.25,0);
\draw[draw=blue,arrows=->](1.5,1.5)--(2.25,1.5);
\draw[draw=blue,arrows=->](0,1.5)--(.75,1.5);
\draw[draw=green,arrows=->](0,0)--(-.525,-.525);
\draw[draw=green,arrows=->](1.5,0)--(2.025,-.525);
\draw[draw=green,arrows=->](1.5,1.5)--(2.025,2.025);
\draw[draw=green,arrows=->](0,1.5)--(-.525,2.025);
\filldraw(0,0) circle (2pt);
\filldraw(1.5,0) circle (2pt);
\filldraw(1.5,1.5) circle (2pt);
\filldraw(0,1.5) circle (2pt);
\draw[draw=black,arrows=->](.75,.75)--(1.125,.75);
\draw[draw=black,arrows=->](.75,.75)--(.75,1.125);
\draw[draw=black,arrows=->](-.75,.45)--(-1.125,.45);
\draw[draw=black,arrows=->](-.75,.45)--(-.75,.075);
\node at (.95,.5) {\small $O_b$ \normalsize};\node at (-.75,.7) {\small $O_g$ \normalsize};\node at (.75,-.2) {\small $z_1^*$ \normalsize};\node at (1.75,.75) {\small $z_2^*$ \normalsize};\node at (.45,.75) {\small $z_3^*$ \normalsize};\node at (.75,1.7) {\small $z_4^*$ \normalsize};\node at (-.25,.75) {\small $z_5^*$ \normalsize};\end{tikzpicture}
\caption{The velocity of the agents at the desired shape is the linear combination of the unit vectors from their associated relative positions (black dashed lines) given by (\ref{eq: uim}). The common velocity $^bv_{p_c}^*$ of the centroid is marked in blue and it is given by $\mu_v$ and $\tilde\mu_v$. The rotational velocity about the centroid defining $^b\omega^*$ is marked in red and is given $\mu_\omega$ and $\tilde\mu_\omega$. The velocity responsible for scaling the formation has been marked in green and it is given by $\mu_s$ and $\tilde\mu_s$. Note that these velocities are constant with respect to $O_b$.}
\label{fig: mis}
\end{figure}

\subsection{Design of the controller for precise motion and changing of the formation}

Here by precise scaling we mean the control of agents such that the inter-distances follow a time varying $d(t)$ but with the shape in $\mathcal{Z}_\mathcal{S}$. More precisely, we set $e_k(t) = ||z_k(t)|| - d_k(t)$ for $k\in\{1, \dots, |\mathcal{E}|\}$ with $\mathcal{Z}\in\mathcal{Z}_\mathcal{S}$ in a family of infinitesimally and minimally congruent rigid shapes. For simplicity we set the following relation in the edge $\mathcal{E}_k$
\begin{equation}
d_k(t) = s(t)d^*_k + d^*_k,
\label{eq: dt}
\end{equation}
where $s(t)\in\R$ is a time varying scaling signal which is assumed to be at least $C^1$ and $d_k^*$ is defined for a particular $\mathcal{Z}$. We remark here that the form used in (\ref{eq: dt}) is for convenience of design. One can of course choose $d_k(t) = s(t)d^*_k$. Without loss of generality, we assume that $s(0) = 0$. Obviously, for well-posedness we also impose that $s(t)$ is defined properly such that $d_k(t) > 0$ for all $t$ and $k\in\{1, \dots, |\mathcal{E}|\}$. 

It is clear that the desired linear speed $||^bv_{p_c}^*||$ and the desired angular speed $||^b\omega^*||$ are related to the norms of $\mu_v, \tilde\mu_v$ and $\mu_\omega, \tilde\mu_\omega$ respectively. It can also be easily checked that the speed $\frac{\mathrm{d}}{\mathrm{dt}}d_k(t)$ is related to the norms of $\mu_s$ and $\tilde\mu_s$.

We derive the dynamics of $z$ and $e$ from (\ref{eq: umis}) but consider the time varying desired distances
\begin{align}
	\dot z &= -c\overline B^T\overline BD_{\hat z}e + \overline B^T\overline A(\mu,\tilde\mu)\hat z \\
	\dot e &= -cD^T_{\hat z}\overline B^T\overline BD_{\hat z}e + D^T_{\hat z}\overline B^T\overline A(\mu,\tilde\mu)\hat z - \dot d, \label{eq: es}
\end{align}
where we have rewritten $e$ as the stacked vector of $e_k(t) = ||z_k(t)|| - d_k(t)$ for $k\in\{1, \dots, |\mathcal{E}|\}$, and $d$ is the stacked vector of $d_k(t)$ also for $k\in\{1, \dots, |\mathcal{E}|\}$.

In a similar way as in (\ref{eq: Z}) and (\ref{eq: E}), in order to compensate $\dot d$ in (\ref{eq: es}) we impose the following condition for keeping invariant $e$ for $e=\mathbf{0}$, i.e. the formation shape is always in $\mathcal{Z}_s$
\begin{equation}
\dot d = D^T_{^b\hat z^*}\overline B^T\begin{bmatrix}\overline S_1D_{^b\hat z^*} & \overline S_2D_{^b\hat z^*} \end{bmatrix}\begin{bmatrix}\mu \\ \tilde\mu \end{bmatrix}, \label{eq: ddot}
\end{equation}
so that the last two terms of the right hand side of (\ref{eq: es}) is zero when $e=0$. Note that the solution to (\ref{eq: ddot}) for $\mu$ and $\tilde\mu$ includes the spaces $\mathcal{U}$ and $\mathcal{W}$. Therefore the distributed changing parameters that we are looking for scaling the desired shape with a desired scaling speed are $[\begin{smallmatrix}\mu_s \\ \tilde\mu_s\end{smallmatrix}]\in\mathcal{S}$ such that (\ref{eq: ddot}) holds.

For the constant growing case, i.e. $s(t) = st$, where $s\in\R^+$ is a common constant scaling speed among all the agents, we have that $\dot d = \frac{\mathrm{d}s}{\mathrm{d}t}d^* = sd^*$ and therefore the solution of (\ref{eq: ddot}) gives constant $\mu_s$ and $\tilde\mu_s$. 
Considering the periodic scaling case we have that $s(t) = s\sin(\omega t)$, which obviously satisfies $\dot d = \frac{\mathrm{d}s}{\mathrm{d}t}d^* = s \omega \cos(\omega t)d^*$. Therefore the changing parameters $\mu_s$ and $\tilde\mu_s$ for the periodic case are the same as we have seen previously calculated for the constant growing case but multiplied by the periodic signal $\omega \cos(\omega t)$, which is obviously independent of the actual shape.

\section{Stability analysis}
\label{sec: sta}
Before presenting the main result, we need to show first that the error system in (\ref{eq: es}) is an autonomous system. Indeed, the second term at the right hand side of (\ref{eq: es}) depends on the dot products of the form $\hat z_i^T\hat z_j$ for $i, j\in\{1, \dots, |\mathcal{E}|\}$. It has been shown in \cite{MouMorseBelSunAnd15} that all the scalar products $z^T_iz_j$ for $i, j\in\{1, \dots, |\mathcal{E}|\}$ can be written as smooth functions of the inter-distances $||z_k||$ for $k \in\{1, \dots, |\mathcal{E}|\}$. Since the errors $e_k$ for $k \in\{1, \dots, |\mathcal{E}|\}$ are functions of only the inter-distances $||z_k||$ and $\hat z_k$ for $k \in\{1, \dots, |\mathcal{E}|\}$, we have that
\begin{equation}
	\hat z_i^T\hat z_j = g_{ij}(e), \quad i, j\in\{1, \dots, |\mathcal{E}|\},
\end{equation}
where $g_{ij}$ is a local smooth function around the shape $z\in\mathcal{Z}_s$. Note that when $z\in\mathcal{Z}_s$, the second and third terms on the right hand side of (\ref{eq: es}) vanish because of (\ref{eq: ddot}), therefore we can write the following local function
\begin{equation}
	f(e) = D^T_{\hat z}\overline B^T\overline A(\mu,\tilde\mu)\hat z - \dot d, \quad f(\mathbf{0}) = \mathbf{0} \iff z\in\mathcal{Z}_s.
	\label{eq: fs}
\end{equation}
Employing the same argument, the matrix in the first term of the right hand side of (\ref{eq: es}) can be rewritten as
\begin{equation}
	Q(e) = D^T_{\hat z}\overline B^T\overline BD_{\hat z}, \label{eq: Q}
\end{equation}
where it has been shown in \cite{MaJaCa15} that $Q(\mathbf{0})$ with $z\in\mathcal{Z}_s$ is positive definite.

\begin{theorem}
\label{th: main}
Consider the distributed parameters $[\begin{smallmatrix}\mu_v \\ \tilde\mu_v\end{smallmatrix}]$, $[\begin{smallmatrix}\mu_\omega \\ \tilde\mu_\omega\end{smallmatrix}]$ and $[\begin{smallmatrix}\mu_s \\ \tilde\mu_s\end{smallmatrix}]$ belonging to the spaces (\ref{eq: U}), (\ref{eq: W}) and (\ref{eq: S}) respectively.
Then, there exist constants $\rho, c^*>0$ such that the origin of the system (\ref{eq: es}), corresponding to $z\in\mathcal{Z}_s$ with time-varying $s(t)$ as in (\ref{eq: dt}), is locally exponentially stable for all $c\geq c^*$ in the compact set $\mathcal{Q}\dfb\{e:||e||^2\leq \rho\}$. In particular, the formation will converge exponentially fast to the time-varying shape defined by $\mathcal{Z}_s$ with the speed $\frac{\mathrm{d}s(t)}{\mathrm{d}t}$ satisfying (\ref{eq: ddot}) and the agents' velocities
	\begin{equation}
		^b\dot p_i(t) \to {^b\dot p_i^*}, \, t\to\infty, \, i\in\{1, \dots,|\mathcal{V}|\},
	\end{equation}
where the ${^b\dot p_i^*}$'s are given by the desired $^b\dot p_c^*$, $^b\omega$ and $\frac{\mathrm{d}s(t)}{\mathrm{d}t}$, that are determined by $[\begin{smallmatrix}\mu_v \\ \tilde\mu_v\end{smallmatrix}]$, $[\begin{smallmatrix}\mu_\omega \\ \tilde\mu_\omega\end{smallmatrix}]$ and $[\begin{smallmatrix}\mu_s \\ \tilde\mu_s\end{smallmatrix}]$.
\end{theorem}
\begin{proof}
Consider the following candidate Lyapunov function
\begin{equation}
	V = \frac{1}{2}||e||^2,
\end{equation}
whose time derivative satisfies
\begin{equation}
	\frac{\mathrm{d}V}{\mathrm{d}t} = e^T\dot e = -ce^TQ(e)e+e^Tf(e),
\end{equation}
	with $f(e)$ and $Q(e)$ as in (\ref{eq: fs}) and (\ref{eq: Q}) respectively. We have that $Q(\mathbf{0})$ is positive definite and that in a neighborhood of $\mathcal{Z}_s$ the formation is still infinitesimally and minimally rigid, therefore $Q(e)$ is positive definite in the compact set $\mathcal{Q}$ for some small positive $\rho$. Furthermore, $f(e)$ is locally Lipschitz in the compact set $\mathcal{Q}$ and $f(\mathbf{0}) = \mathbf{0}$ with $z\in\mathcal{Z}_s$, therefore there exists a constant $q\in\R^+$ such that
\begin{equation}
	\frac{\mathrm{d}V}{\mathrm{d}t} \leq (-c\lambda_{\text{min}}+q)||e||^2,
\end{equation}
	where $\lambda_{\text{min}}$ is the minimum eigenvalue of $Q(e)$ in $\mathcal{Q}$. Thus if one chooses $c\geq c^*>\frac{q}{\lambda_{\text{min}}}$, then the exponential stability of the origin of (\ref{eq: es}) follows from Theorem 4.10 from \cite{khalil1996nonlinear} for non-autonomous systems. Therefore we have that the formation shape converges exponentially to $\mathcal{Z}_s$.

	Now we substitute $e(t)\to \mathbf{0}$ and $\hat z(t) \to \mathcal{Z}_s$ as $t$ goes to infinity into (\ref{eq: umis}) and (\ref{eq: pdyn}), which gives us
\begin{equation}
	\dot p(t) - \bar A(\mu, \tilde\mu)\hat z(t) \to \mathbf{0}, \quad t\to\infty.
\end{equation}
In other words, the velocity of the formation converges exponentially fast to the desired velocity given as a superposition of $^bv^*_{p_c}$ and $^b\omega^*$  with the scaling speed $\frac{\mathrm{d}s(t)}{\mathrm{d}t}$ satisfying (\ref{eq: ddot}).
\end{proof}
\begin{remark}
The magnitude of the positive constant $\rho$ only depends on the desired shape, i.e. if one chooses a shape where all the agents are far away from each other and far away from a collinear (2D) or coplanar (3D) configuration, then one should expect a bigger $\rho$ for such desired shape than for the one that does not meet such requirements. In some sense $\rho$ is measuring (in a conservative way) how the desired formation can be distorted without falling into a degenerated configuration. Examples about how to compute $c^*$ and $q$ can be found in the PhD thesis \cite{thesis}.
\end{remark}

\section{Simulation results}
\label{sec: exp}
In this section \footnote{Video footage from actual mobile robots can be found at www.youtube.com/c/HectorGarciadeMarina with their explanations in \cite{thesis}.} we validate the correctness of Theorem \ref{th: main}. We have four agents with a scale-free regular square as the prescribed shape. The objective of this simulation is to design the distributed motion-changing parameters $\mu$ and $\tilde\mu$ in the control law (\ref{eq: umis}) such that the square spins around its centroid and at the same time we vary periodically the scale of the square. We define the sensing topology of the agents by
$	B = \bigg[\begin{smallmatrix}
		1 & 0 & -1 & 0 & -1 \\ -1& 1& 0& 0& 0 \\ 0& -1& 1& -1& 0 \\ 0& 0& 0& 1& 1
	\end{smallmatrix}\bigg]$
and define the regular square that we will periodically scale with side-length $d_1^* = 15$ pixels. In order to induce the spinning motion we design the following $\mu_\omega$ and $\tilde\mu_\omega$ satisfying (\ref{eq: W})
\begin{equation}
\mu_\omega = [\begin{smallmatrix}-w & -w & 0 & w & -w\end{smallmatrix}]^T, \,
\tilde\mu_\omega = [\begin{smallmatrix}-w & -w & 0 & w & -w\end{smallmatrix}]^T, \label{eq: misW}
\end{equation}
with $w = 1$. We want to vary periodically the size of the square following the desired time-varying distances
\begin{equation}
	d_i(t) = d_i^* + 2hd_i^*\sin(\omega_s t),\, i=\{1,\dots,5\},
\label{eq: dsqp}
\end{equation}
where one can deduce that $s(t) = 2h\sin(\omega_s t)$ and we set $h = 2$ and $\omega_s = 1.5$ rads/sec. The desired $\mu_s$ and $\tilde\mu_s$ satisfying (\ref{eq: S}) and (\ref{eq: ddot}) are
\begin{equation}
	\begin{smallmatrix}[\mu_s(t)^T & \tilde\mu_s(t)^T]\end{smallmatrix} = h\omega_s\cos(\omega_s t)[\begin{smallmatrix}1 & 1 & 0 & 1 & 1 & -1 & -1 & 0 & -1 & -1\end{smallmatrix}]. \label{eq: misS}
\end{equation}
Finally we choose $c = 5$ for (\ref{eq: umis}), which in numerical checking is much smaller than the conservative gain in Theorem \ref{th: main}. The numerical results are shown in Figure \ref{fig: per}.

\begin{figure}
\centering
\includegraphics[width=0.215\textwidth]{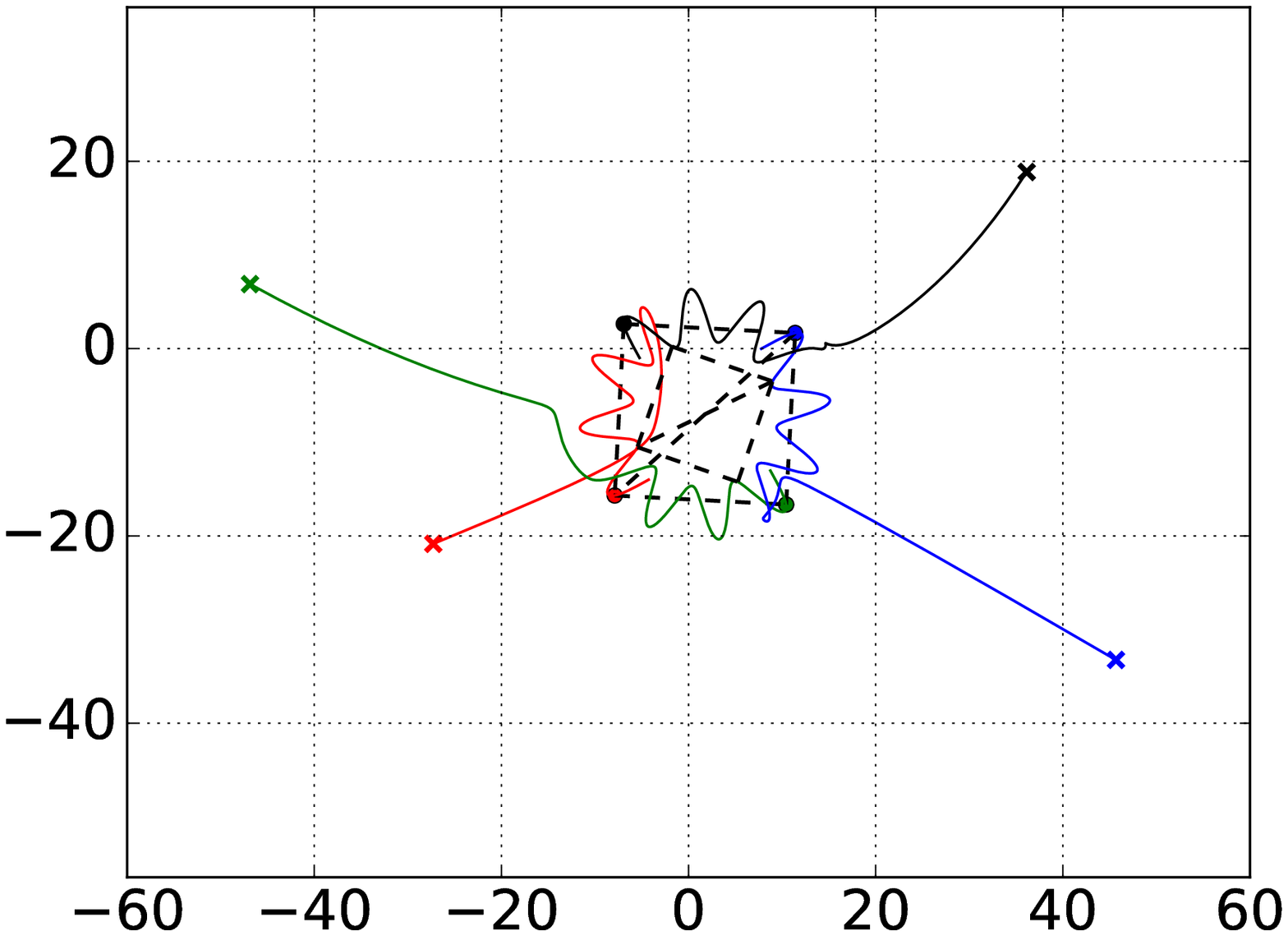}
\includegraphics[width=0.215\textwidth]{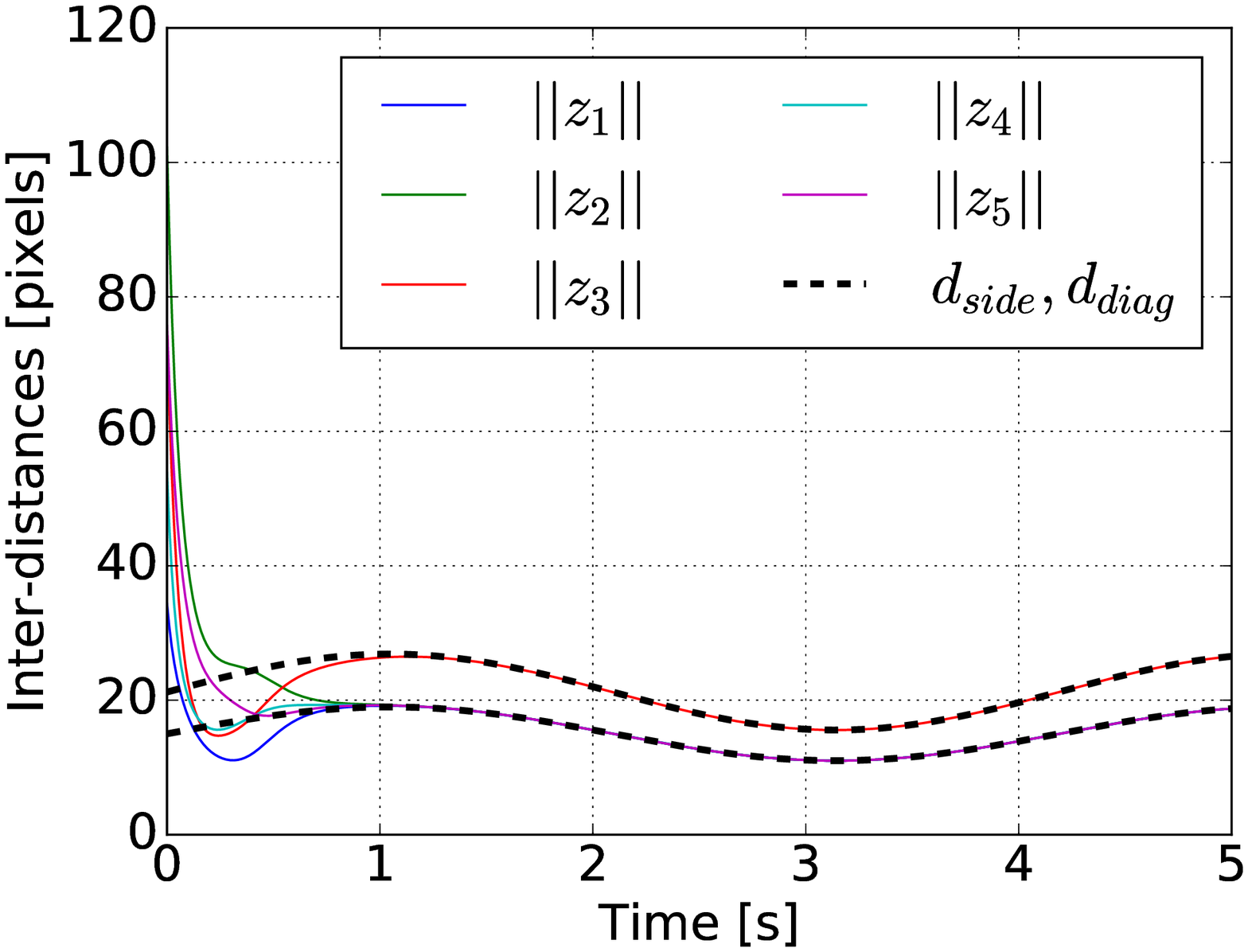}
\caption{The left plot shows the trajectories of the agents with the `x' denoting the initial points. The dashed lines show two different scales of the prescribes square. The right plot shows the evolution of the inter-distances.}
\label{fig: per}
\end{figure}

\section{Conclusions}
In this paper we have modified the popular distance-based controller by adding distributed parameters at their prescribed inter-distances in order to control the steady-state motion while at the same time controlling precisely the scaling rate of the formation. This approach is compatible with higher order agent dynamics \cite{MaJaCa16} and is applicable to the target enclosing and tracking problem. For the periodic scaling, future work includes the addition of estimators based on the internal model principle in order not to require all the scaling signals $s(t)$ to have the same phase at the starting time.
\bibliographystyle{IEEEtran}
\bibliography{hector_ref}

\begin{thebibliography}{10}
\providecommand{\url}[1]{#1}
\csname url@samestyle\endcsname
\providecommand{\newblock}{\relax}
\providecommand{\bibinfo}[2]{#2}
\providecommand{\BIBentrySTDinterwordspacing}{\spaceskip=0pt\relax}
\providecommand{\BIBentryALTinterwordstretchfactor}{4}
\providecommand{\BIBentryALTinterwordspacing}{\spaceskip=\fontdimen2\font plus
\BIBentryALTinterwordstretchfactor\fontdimen3\font minus
  \fontdimen4\font\relax}
\providecommand{\BIBforeignlanguage}[2]{{%
\expandafter\ifx\csname l@#1\endcsname\relax
\typeout{** WARNING: IEEEtran.bst: No hyphenation pattern has been}%
\typeout{** loaded for the language `#1'. Using the pattern for}%
\typeout{** the default language instead.}%
\else
\language=\csname l@#1\endcsname
\fi
#2}}
\providecommand{\BIBdecl}{\relax}
\BIBdecl

\bibitem{sheng2006distributed}
W.~Sheng, Q.~Yang, J.~Tan, and N.~Xi, ``Distributed multi-robot coordination in
  area exploration,'' \emph{Robotics and Autonomous Systems}, vol.~54, no.~12,
  pp. 945--955, 2006.

\bibitem{KrBrFr08}
L.~Krick, M.~E. Broucke, and B.~A. Francis, ``Stabilization of infinitesimally
  rigid formations of multi-robot networks,'' \emph{International Journal of
  Control}, vol.~82, pp. 423--439, 2009.

\bibitem{AnYuFiHe08}
B.~D.~O. Anderson, C.~Yu, B.~Fidan, and J.~Hendrickx, ``Rigid graph control
  architectures for autonomous formations,'' \emph{IEEE Control Systems
  Magazine}, vol.~28, pp. 48--63, 2008.

\bibitem{cao2007controlling}
M.~Cao, A.~Morse, C.~Yu, B.~Anderson, and S.~Dasgupta, ``Controlling a
  triangular formation of mobile autonomous agents,'' in \emph{Decision and
  Control, 2007 46th IEEE Conference on}.\hskip 1em plus 0.5em minus
  0.4em\relax IEEE, 2007, pp. 3603--3608.

\bibitem{SuLiAn15}
Z.~Sun, Q.~Liu, C.~Yu, and B.~Anderson, ``Generalized controllers for rigid
  formation stabilization with application to event-based controller design,''
  in \emph{Proc. of the European Control Conference (ECC'15)}, 2015, pp.
  217--222.

\bibitem{MarCaoJa15}
H.~Garcia~de Marina, M.~Cao, and B.~Jayawardhana, ``Controlling rigid
  formations of mobile agents under inconsistent measurements,''
  \emph{Robotics, IEEE Transactions on}, vol.~31, no.~1, pp. 31--39, Feb 2015.

\bibitem{zhao2014bearing}
S.~Zhao and D.~Zelazo, ``Bearing rigidity and almost global bearing-only
  formation stabilization,'' \emph{Automatic Control, IEEE Transactions on, to
  appear.}, 2016.

\bibitem{zhao2015translational}
------, ``Translational and scaling formation maneuver control via a
  bearing-based approach,'' \emph{Control of Network Systems, IEEE Transactions
  on, to appear.}, 2016.

\bibitem{MouMorseBelSunAnd15}
S.~Mou, A.~S. Morse, A.~Belabbas, Z.~Sun, and B.~Anderson, ``Undirected rigid
  formations are problematic,'' \emph{IEEE Transactions on Automatic Control,
  to appear.}, 2016.

\bibitem{sun2014formation}
Z.~Sun, S.~Mou, B.~D. Anderson, and A.~S. Morse, ``Formation movements in
  minimally rigid formation control with mismatched mutual distances,''
  \emph{Proc. of the 53rd IEEE Conference on Decision and Control (CDC 2014)},
  pp. 6161--6166, 2014.

\bibitem{MaJaCa15}
H.~Garcia~de Marina, B.~Jayawardhana, and M.~Cao, ``Distributed rotational and
  translational maneuvering of rigid formations and their applications,''
  \emph{IEEE Transactions on Robotics}, vol.~32, no.~3, pp. 684--697, June
  2016.

\bibitem{6862558}
D.~Zelazo, A.~Franchi, and P.~Giordano, ``Rigidity theory in se(2) for unscaled
  relative position estimation using only bearing measurements,'' in
  \emph{Control Conference (ECC), 2014 European}, June 2014, pp. 2703--2708.

\bibitem{khalil1996nonlinear}
H.~K. Khalil and J.~Grizzle, \emph{Nonlinear systems}.\hskip 1em plus 0.5em
  minus 0.4em\relax Prentice hall New Jersey, 1996, vol.~3.

\bibitem{thesis}
H.~{Garcia de Marina Peinado}, ``Distributed formation control for autonomous
  robots,'' Ph.D. dissertation, 2016.

\bibitem{MaJaCa16}
H.~Garcia~de Marina, B.~Jayawardhana, and M.~Cao, ``Taming inter-distances
  mismatches for formation-motion control of rigid formations in second-order
  agents,'' \emph{Automatic Control, IEEE Transactions on,}, 2016, submitted.

\end{thebibliography}

\end{document}